\documentclass[a4paper,12pt]{article}
\usepackage{amsmath}
\usepackage{amsthm}
\usepackage{amsfonts}
\usepackage{mathrsfs}
\usepackage{graphicx}
\usepackage{hyperref}

\newtheorem{theorem}{Theorem}
\newtheorem{lemma}{Lemma}

\numberwithin{equation}{section}

\hypersetup{ pdftitle={Draft},
pdfauthor={Fiki T. Akbar, Bobby E. Gunara}, 
pdfkeywords={Local Esistence, Dyonic Black Hole}, 
bookmarksnumbered, pdfstartview={FitH}, urlcolor=blue, }

\begin{document}
\pagestyle{plain}




\title{\LARGE\textbf{Existence  of Static Dyonic Black Holes in $4d$ $N = 1$ Supergravity With Finite Energy}}

\author{{Fiki T. Akbar$^{\sharp}$ and Bobby E. Gunara}$^{\flat,\sharp}$ \\ \\
$^{\flat}$\textit{\small Indonesian Center for Theoretical and
Mathematical Physics (ICTMP)}
\\ {\small and} \\
$^{\sharp}$\textit{\small Theoretical Physics Laboratory}\\
\textit{\small Theoretical High Energy Physics and Instrumentation Research Group,}\\
\textit{\small Faculty of Mathematics and Natural Sciences,}\\
\textit{\small Institut Teknologi Bandung}\\
\textit{\small Jl. Ganesha no. 10 Bandung, Indonesia, 40132}\\
\small email: ftakbar@fi.itb.ac.id, bobby@fi.itb.ac.id}

\date{}

\maketitle




\begin{abstract}

We prove the existence and the uniqueness of the static dyonic
black holes in four dimensional $N=1$ supergravity theory coupled
vector and scalar multiplets.
 We set the near-horizon geometry to be a product of two Einstein surfaces, whereas the asymptotic geometry has to be a space of constant scalar curvature.
 Using these data, we show that there exist a unique solution for scalar fields which interpolates these regions.

\end{abstract}




\section{Introduction}
\label{sec:intro}

Four dimensional static spacetimes, which are  closely related to
the spherical symmetric black holes, have been intensely studied
because it may be a simple model and useful to test an extension
theory of Einstein's general relativity.  For example,  a
supersymmetric extension of Einstein's general relativity called
$N=1$ supergravity is of interest to consider since in the rigid
limit it may provide a unification theory.\\
\indent Recently, some authors considered a class of extremal
spherical symmetric black hole solutions of four dimensional $N=1$
supergravity coupled to vector and chiral multiplets with the
scalar potential turned on \cite{Gunara:2013IJMPA}. These black
holes are non-supersymmetric. The boundaries of the black holes
are set to be spaces of constant scalar curvature. Near the
horizon, the geometry becomes a product space formed by two
Einstein surfaces, namely two-anti de Sitter surface ($AdS_2$) and
two-sphere $S^2$. While, in the asymptotic region the geometry
turns to a space of constant scalar curvature which is
not Einstein.\\
\indent At the boundaries  the complex scalars are frozen and
regular. If the value of the scalars in both regions coincides,
then both  asymptotic and near-horizon geometries may correspond
to each other. This can be achieved if the Arnowitt-Deser-Misner (ADM) mass is extremized
which implies  that  the scalar charges vanish \cite{Gunara:2013IJMPA}.\\
\indent In this paper we prove locally and globally the existence
of static dyonic black holes of $N=1$ local supersymmetry
(supergravity) with general coupling of vector and scalar
multiplets in four dimensions. These black holes are generally
non-extremal and non-supersymmetric. We use the same setting for
the geometry of the boundaries of the black holes as the previous
case mentioned above. Our result here also covers the  non-supersymmetric case of $N=1$ theory \cite{Andrianopoli:2007} and could be applied to the  supersymmetric black holes in $N=2$ theory \cite{Ferrara:1995ih, Ferrara:1996prd}.\\
\indent To prove the local existence and the uniqueness, we use a
contraction mapping theorem to the complex scalar field equations
of motions.
 By defining an integral operator, the existence and the uniqueness  can be achieved by showing
 that the non-linear parts satisfy the local Lipshitz condition. This can be done by taking several assumptions in the following:  First, both K\"ahler potential and Christoffel symbol are bounded above by $U(n)$ symmetric K\"ahler potential and its Christoffel symbol,
respectively. These estimates can be used to eliminate the quantity associated with K\"ahler metric in our analysis. Second, the potentials, namely the scalar potential and the black hole potential have to be at least $C^2$-functions and their derivative
is locally Lipshitz function. Finally, the
spacetime metric is at least a $C^2$-function.\\
\indent The second step is to prove the global existence by
showing that the energy functional is bounded above a positive
constant between the boundaries. We also take that all functions
namely the scalar fields, the potentials, and the spacetime metric
have to be at least a $C^2$-function. After some computations, we
find that the scalar potential should be vanished in the
asymptotic region.\\
\indent The organization of this paper can be mentioned as
follows. In section \ref{sec:sugra} we review shortly $N = 1$
supergravity coupled with vector and chiral multiplets in four
dimensions. We describe some aspects of static spacetimes whose
boundaries are spaces of constant scalar curvature in section
\ref{sec:setup}. In section \ref{sec:exist} we prove the local
existence of radius dependent scalar fields which follows that
such static black holes do exist in any compact interval. Finally,
we employ an energy functional analysis to show the solutions
could be globally defined in section \ref{sec:energy}.

\section{$N = 1$ Supergravity Coupled with Vector and Chiral Multiplets}
\label{sec:sugra}

In this section, we give a short description about four
dimensional $N = 1$ supergravity coupled with vector and chiral
multiplets. Here, we only write some useful terms  for our
analysis in the paper. For an excellent review, interested reader
can further consult, for example,
\cite{D'Auria:2001kv,Andrianopoli:2001zh}.

The theory consists of a gravitational multiplet  $(e_{\mu}^\Lambda, \psi_\mu)$, $n_C$ chiral multiplet $(\phi^{i},\chi^{i})$ and $n_V$ vector multiplets $(A^{a}_{\mu}, \lambda^{a})$ where Latin alphabets $a,b =1,...,n_v$, and $i,j=1,...,n_c$ show the number of multiplets and the Greek alphabets $\mu,\nu =0,...,3$ and $\Lambda,\Sigma = 0,...,3$ show the spacetime and tangent space index respectively. Here $e_{\mu}^\Lambda$, $A^{a}_{\mu}$, and $\phi$ are a vierbein, a
gauge field, and a complex scalar, respectively, while $\psi_\mu$,
$\lambda^{a}$, and $\chi^{a}$ are the fermion fields.

Furthermore, the complex scalar fields span a $2n_c$-dimensional Hodge-K\"ahler manifold $\mathcal{M}^{2n_c}$ endowed with metric $g_{i\bar{j}} = \partial_{i}\partial_{\bar{j}}\:K(\phi,\bar{\phi})$ and a $U(1)$ connection is defined by $Q  \equiv  - \left( K_i \,d\phi^i - K_{\bar{j}}\, d\bar{\phi}^{\bar{j}}\right)$, where $K(\phi,\bar{\phi})$ is a real function called K\"ahler potential. Moreover, there exists a set of holomorphic functions, $(f_{ab}(\phi),W(\phi))$, which are the gauge couplings and the superpotential respectively. The bosonic sector of Lagrangian is given by,
\begin{eqnarray}
\frac{1}{\sqrt{-g}} {\mathcal{L}} & = & - \frac{1}{2} R + h_{ab}\,{\mathcal{F}}_{\mu\nu}^{a}{\mathcal{F}}^{b\mu\nu}
+  k_{ab}\,
{\mathcal{F}}_{\mu\nu}^{a}\widetilde{\mathcal{F}}^{b\mu\nu}
\nonumber\\
 && + \, g_{i\bar{j}}(\phi^{i}, \bar{\phi}^{\bar{j}})\,
\partial_{\mu} \phi^{i} \,
\partial^{\mu}\bar{\phi}^{\bar{j}} - V(\phi,\bar{\phi})\;, \label{LagrangianSupergravity}
\end{eqnarray}
where $V_{S}$ is the real function called the scalar potential which given by,
\begin{equation}
V(\phi,\bar{\phi}) = e^{ K}\left(g^{i\bar{j}} \, \nabla_i W\,
 \bar{\nabla}_{\bar{j}} \bar{W}
 - 3  W \bar{W} \right)  \;. \label{ScalarPotential}
\end{equation}
The quantity $R$ is the Ricci scalar of
four dimensional spacetime, whereas
${\mathcal{F}}_{\mu\nu}^{a}$ is an Abelian field strength of
$A_{\mu}^{a}$, and
$\widetilde{{\mathcal{F}}}_{\mu\nu}^{a}$ is a Hodge dual of
${\mathcal{F}}_{\mu\nu}^{a}$. The function $h_{ab}$ and $k_{ab}$ are the real and imaginary part of $f_{ab}$ respectively and $\nabla_i W\equiv
\partial_i W +  K_i W$ is a covariant derivative in internal manifold, $\mathcal{M}^{2n_c}$. It is worth to mention that Lagrangian (\ref{LagrangianSupergravity}) is invariant under a set of supersymmetry transformations of the fields discussed in \cite{D'Auria:2001kv,Andrianopoli:2001zh}.

The bosonic equations of motions can be obtained by varying the Lagrangian (\ref{LagrangianSupergravity}) with respect to $g_{\mu\nu}$, $A^{a}_{\mu}$, and $\phi^{i}$, and by setting all the fermions are vanish at the level of equation of motion. The first equation is namely the Einstein field equation,
\begin{eqnarray}
 R_{\mu\nu} - \frac{1}{2} g_{\mu\nu} R  &=& g_{i\bar{j}} \,
 (\partial_{\mu}\phi^i \partial_{\nu}\bar{\phi}^{\bar{j}}
 + \partial_{\nu}\phi^i \partial_{\mu}\bar{\phi}^{\bar{j}}) - g_{i\bar{j}} \, g_{\mu\nu}
\, \partial_{\rho}\phi^i \partial^{\rho}\bar{\phi}^{\bar{j}} \nonumber\\
&& + \, 4 {h}_{ab}
\,{\mathcal{F}}_{\mu\rho}^{a}{\mathcal{F}}^{b}_{\nu\sigma}
g^{\rho\sigma} - g_{\mu\nu} {h}_{ab}
\,{\mathcal{F}}_{\rho\sigma}^{a}{\mathcal{F}}^{b|\rho\sigma}
 + g_{\mu\nu} V \;. \label{Einsteineq}
\end{eqnarray}
The second equation is the equation for gauge fields,
\begin{equation}
 \partial_{\nu} \left(\varepsilon^{\mu\nu\rho\sigma} \sqrt{-g}
 \, {\mathcal{G}}_{a|\rho\sigma}\right) = 0\;,
\label{SugragaugeEOM}
\end{equation}
where
\begin{equation}
{\mathcal{G}}_{a|\rho\sigma}  \equiv {k}_{ab} {\mathcal{F}}^{b}_{\rho\sigma} -
{h}_{ab}
\widetilde{\mathcal{F}}^{b}_{\rho\sigma} \;,
\label{magnetfield}
\end{equation}
is an electric field strengths tensor. Finally, we have the scalar field equations,
\begin{equation}
 \frac{g_{i\bar{j}}}{\sqrt{-g}} \, \partial_{\mu} \left( \sqrt{-g} \,g^{\mu\nu}
 \partial_{\nu}\bar{\phi}^{\bar{j}}
 \right) + \bar{\partial}_{\bar{k}} g_{i\bar{j}} \, \partial_{\nu}\bar{\phi}^{\bar{j}}
 \partial^{\nu}\bar{\phi}^{\bar{k}} =  \partial_i {h}_{ab}\,{\mathcal{F}}_{\mu\nu}^{a}{\mathcal{F}}^{b|\mu\nu}
-  \partial_i{k}_{ab}\,
{\mathcal{F}}_{\mu\nu}^{a}\widetilde{\mathcal{F}}^{b|\mu\nu}
- \partial_i V \;, \label{SugrascalarEOM}
\end{equation}
with $g \equiv {\mathrm{det}}(g_{\mu\nu})$. In addition, we have a Bianchi identity
\begin{equation}
 \partial_{\nu} \left(\varepsilon^{\mu\nu\rho\sigma} \sqrt{-g}
 \, {\mathcal{F}}^{a}_{\rho\sigma}\right) = 0\;,
\label{BianchiId}
\end{equation}
comes from the definition of  ${\mathcal{F}}^{a}_{\rho\sigma}$.

\section{General Setting: Static Spacetimes}
\label{sec:setup}

 In this section we set a scenario  as follows.
First of all, we consider the following ansatz for spacetime
metric,
\begin{equation}
ds^2 = -e^{A(r)}\, dt^2 + e^{B(r)}\, dr^2 + e^{C(r)}\, (d\theta^2
+ {\mathrm{sin}}^2\theta \, d\phi^2) \ , \label{metricans}
\end{equation}
which is static and has a spherical symmetry. The functions $A(r)$, $B(r)$, and $C(r)$ are a smooth real functions.
 Among the functions $A(r)$, $B(r)$, and $C(r)$, only two of them are independent since one can redefine the radial coordinate $r$ to absorb one of them.\\
\indent The solution for the gauge field
equations,(\ref{SugragaugeEOM}), and the Bianchi identity,
(\ref{BianchiId}), can be obtained by simply taken the case where
the only non zero component of the field strength tensor are
${\mathcal{F}}^{a}_{01}(r)$ are ${\mathcal{F}}^{a}_{23}(\theta)$.
The solutions are,
\begin{eqnarray}
{\mathcal{F}}^{a}_{01} &=&
\frac{1}{2} \, e^{\frac{1}{2}(A + B)-C}\,
({h}^{-1})^{ab}({k}_{bc}\,
g^{c} - q_{b}) \;,\nonumber\\
{\mathcal{F}}^{a}_{23} &=& - \frac{1}{2} \, g^{a} \,
{\mathrm{sin}}\theta \;, \label{solgaugeEOM}
\end{eqnarray}
with $q_{a}$ and $g^{a}$ are electric and magnetic charges,
respectively \cite{Belluci:2008prd}.
Then, by assuming that the scalar field is function of radial
coordinate only, the scalar field equations,
(\ref{SugrascalarEOM}), can be written as
\begin{equation}
 g_{i\bar{j}} \, \bar{\phi}^{\bar{j}}{''} + \bar{\partial}_{\bar{k}}
 g_{i\bar{j}} \, \bar{\phi}^{\bar{j}}{'} \bar{\phi}^{\bar{k}}{'}
 + \frac{1}{2} \left(A' - B' + 2C' \right)g_{i\bar{j}} \,
 \bar{\phi}^{\bar{j}}{'} = e^{B} \left( e^{-2C} \partial_i V_{\mathrm{BH}}
 + \partial_i V \right) \;,\label{scalarEOM1}
\end{equation}
with $\nu' \equiv d\nu / dr$ denotes the derivative respect to radial coordinate. The function $V_{\mathrm{BH}}$ is called the black hole potential \cite{Ferrara:1996prd, Ferrara:1997npb} and it has the form
\begin{eqnarray}
V_{\mathrm{BH}} \equiv - \frac{1}{2} \, (g \:\, q
) \; {\mathcal{M}} \; \left(\begin{array}{c}  g \\
 q \end{array} \right)
 \;,
 \label{VBH}
\end{eqnarray}
where
\begin{eqnarray}
{\mathcal{M}} = \left(\begin{array}{cc} {h} +
{k} \,{h}^{-1} \,{k} & -
{k}\, {h}^{-1} \\
- {h}^{-1}\, {k} & {h}^{-1}
\end{array} \right) \;. \label{McalBH}
\end{eqnarray}

Now we can construct a class of solutions of equations (\ref{Einsteineq}) and (\ref{scalarEOM1}) for particular regions, namely near asymptotic and near horizon regions discussed in \cite{Gunara:2013IJMPA}. Around asymptotic region, the scalar fields are frozen, namely
\begin{eqnarray}
\lim_{r \to\infty}\phi^i{'}  &=&  0
 \;,\nonumber\\
\lim_{r \to \infty} \phi^i &=& \phi^i_0  \ , \label{inftycon}
\end{eqnarray}
 which implies that both the black hole potential and the scalar potential become constant. The black hole solutions can be written as,
\begin{equation}
ds^{2} = -\Lambda(r) dt^{2} + \Lambda^{-1}(r)dr^{2} + r^{2}
(d\theta^2
+ {\mathrm{sin}}^2\theta \, d\phi^2) \ , \label{metricsol}
\end{equation}
with
\begin{equation}
\Lambda(r) \equiv 1 + \frac{2\eta}{r} +
\frac{V_{\mathrm{BH}}^{0}}{r^2} - \frac{1}{3}V_{0}r^2\:,
\label{Ar}
\end{equation}
and
\begin{equation}
V_{\mathrm{BH}}^{0} \equiv
V_{\mathrm{BH}}(\phi_{0},\bar{\phi}_{0})\:, \qquad V_{0} \equiv
V(\phi_{0},\bar{\phi}_{0})\ .
\end{equation}
 Since $\Lambda(r)$ is a positive definite
function, the geometries of the black hole solution
(\ref{metricsol}) has a constant scalar curvature which are
neither Einstein nor symmetric space. Furthermore,  the scalar
fields can be written down as,
\begin{eqnarray}
\phi^i{'}  &=&  - \frac{\Sigma^i}{r^2} + \left( P{'}(r)
   \left(g^{i\bar{j}}\, \bar{\partial}_{\bar{j}} V_{\mathrm{BH}}\right)_{\phi_0}
 + Q{'}(r) \left(g^{i\bar{j}} \,\bar{\partial}_{\bar{j}} V\right)_{\phi_0} \right)
 \;,\nonumber\\
\phi^i &=& \phi^i_0 + \frac{\Sigma^i}{r} + \left( P(r)
   \left(g^{i\bar{j}}\,  \bar{\partial}_{\bar{j}} V_{\mathrm{BH}}\right)_{\phi_0}
 + Q(r) \left(g^{i\bar{j}}\, \bar{\partial}_{\bar{j}} V\right)_{\phi_0}
 \right)\;, \label{solscalarEOM2}
\end{eqnarray}
where $\Sigma^i$ are the scalar charges introduced by  \cite{Gibbons:1996prl}. The functions $P(r)$ and $Q(r)$ are
\begin{eqnarray}
P(r) &=& \frac{1}{r}\int \left( \int \frac{A^{-1}}{r^3} \, dr \right) dr \;, \nonumber\\
Q(r) &=& \int \left(\int r \, A^{-1}  \, dr \right) dr \;. \label{PQ}
\label{PQtil}
\end{eqnarray}
In order to evade the singularity in (\ref{PQ}),  the frozen scalars should be critical points of both the black hole and scalar potentials, namely
\begin{eqnarray}
\left( \partial_i V_{\mathrm{BH}}\right)_{\phi_0} &=& 0 \;, \nonumber\\
\left( \partial_i V\right)_{\phi_0}  &=& 0 \;. \label{extremcon1}
\end{eqnarray}

On the other hand, in the near-horizon limit the scalar fields are also frozen with respect to radial coordinates
\begin{eqnarray}
\lim_{r \to r_h} \phi^i & \to & \phi^i_h \ , \nonumber \\
\lim_{r \to r_h} \phi^i{'} & \to & 0 \ , \label{extremcon2}
\end{eqnarray}
and in general $\phi^i_h \ne \phi^i_0$ which can further be viewed as the critical points of the so-called
effective black hole potential $V_{\mathrm{eff}}$ which is a function of both the black hole potential $V_{\mathrm{BH}}$ and the scalar potential $V$
\begin{equation}
 V_{\mathrm{eff}}  \equiv   \frac{1-\sqrt{1-4 V_{\mathrm{BH}}V}}{2
 V}\ .\label{effV}
\end{equation}
Here, the black hole geometries are the product of a two
dimensional surface $M^{1,1}$ and the two-sphere $S^2$ with radius
$r_h \equiv r_h(g, q)$. It is worth mentioning that the near
horizon solution of the scalar fields equation (\ref{scalarEOM1})
has the form
\begin{eqnarray}
  \bar{\phi}^{\bar{j}}{''}
& = & \frac{\ell^{-1}}{(r - r_h)^2} \left( g^{i\bar{j}} \frac{\partial V_{\mathrm{eff}}}{\partial \phi^i} (p_h) \right) \ ,\nonumber\\
  \bar{\phi}^{\bar{j}}
& = &\bar{\phi}^{\bar{j}} _h - \ell^{-1} {\mathrm{ln}} | r - r_h |
\left( g^{i\bar{j}} \frac{\partial V_{\mathrm{eff}}}{\partial \phi^i} (p_h) \right) \:,\label{solscalarEOM3}
\end{eqnarray}
where $p_h \equiv (\phi_h, \bar{\phi}_h)$ and
\begin{equation}
\ell^{-1} = \frac{V^h_{\mathrm{eff}}}{\sqrt{1 -4
V^h_{\mathrm{BH}} V_h}}
 \;, \label{solfieldEOMhor}
\end{equation}
is a constant which shows that $\phi^i_h$ are indeed the critical
points of $V_{\mathrm{eff}}$ in order to have a regular solution
at the region. The Killing vector analysis at the horizon gives us
an evidence that $M^{1,1}$ must be $AdS_2$ \cite{KLR}.

\section{Local Existence}
\label{sec:exist}

In this section we prove the local existence of non-trivial radius dependent solutions of the scalar equations of motions (\ref{scalarEOM1}). This class of solutions interpolates between the two regions, namely the horizon and asymptotic regions.

Let $\mathcal{M}^{2n_c}$ be a $2n_{c}-$dimensional K\"ahler manifold spanned by scalar fields $(\phi^{i},\bar{\phi}^{\bar{i}})$ with K\"ahler potential $K=K(\phi,\bar{\phi})$. In this paper, we consider the case where the K\"ahler potential of $\mathcal{M}^{2n_c}$  is bounded above by a $U(n_{c})$ symmetric K\"ahler potential and satisfies several conditions,
\begin{eqnarray}
K \leq \Phi(|\phi|) \: ,\label{KahlerPotential}\\
\left|\Gamma\right| \leq |\tilde{\Gamma}| \: ,\label{ChristoffelSymbol}
\end{eqnarray}
where $|\phi| =
\left(\delta_{i\bar{j}}\phi^{i}\bar{\phi}^{\bar{j}}\right)^{\frac{1}{2}}$ and $\tilde{\Gamma}$ is the Christoffel symbol of $\tilde{g}$. Then we have the following lemma \cite{AGTZ},
\begin{lemma}
\label{LemmaSIM}
Let $\mathcal{M}^{2n_c}$ be a K\"ahler manifold with Kahler potential $K=K(\phi,\bar{\phi})$. If $\mathcal{M}^{2n_c}$  satisfies (\ref{KahlerPotential}), (\ref{ChristoffelSymbol}) and
\begin{equation}
\left|\frac{{\mathcal F}'}{2|\phi|}\right|  \leq  \epsilon \ ,
\label{ConditionKahler}
\end{equation}
where ${\mathcal F}(|\phi|) \equiv  \frac{1}{4|\phi|^{2}}\left(\Phi''-\frac{\Phi'}{|\phi|}\right)$
with $\Phi' = \partial \Phi/\partial |\phi|$ and $\epsilon$ is a non negative constant, then we have the following estimates
\begin{eqnarray}
\left|K\right| & \leq  & \frac{\epsilon}{6} \left|\phi\right|^{6} +
\frac{C_{1}}{2}\left|\phi\right|^{4} + C_{2}\left|\phi\right|^{2} +
C_{3} \ , \\
\left|\Gamma\right| & \leq & 2\epsilon |\phi|^{3} + C_{1}|\phi| \: .
\end{eqnarray}
\end{lemma}

Let $\phi : I\subset \mathbb{R} \rightarrow {\mathcal{M}}^{2n_c}$ is a curve in ${\mathcal{M}}^{2n_c}$ satisfying differential equations,
\begin{equation}
\bar{\phi}^{\bar{i}}{''} + F(r) \bar{\phi}^{\bar{i}}{'} = - \bar{\Gamma}^{\bar{i}}_{\bar{j}\bar{k}}\bar{\phi}^{\bar{j}}{'}\bar{\phi}^{\bar{k}}{'}  + g^{k\bar{i}}\left(G(r)\partial_{k}V_{BH} + H(r)\partial_{k}V\right) \:, \label{EoM}
\end{equation}
where $F(r) \equiv \frac{1}{2}\left(A'-B'-2C'\right)$, $G(r)
\equiv e^{B-2C}$ and $H(r) \equiv e^{B}$. We assume that the
functions $F(r),G(r)$ and $H(r)$ are at least $C^2$- real
functions. Introducing the new fields,
\begin{equation}
\pi^{i} = \phi^{i}{'}\:, \qquad \bar{\pi}^{\bar{i}} =  \bar{\phi}^{\bar{i}}{'}\:,
\end{equation}
then we can write equation (\ref{EoM}) as a first order equations,
\begin{eqnarray}
\bar{\phi}^{\bar{i}}{'} & = & \bar{\pi}^{\bar{i}}  \:,\nonumber\\
\bar{\pi}^{\bar{i}}{'} + F(r) \bar{\pi}^{\bar{i}} & = & - \bar{\Gamma}^{\bar{i}}_{\bar{j}\bar{k}}\bar{\pi}^{\bar{j}}\bar{\pi}^{\bar{k}}  + g^{k\bar{i}}\left(G(r)\partial_{k}V_{BH} + H(r)\partial_{k}V\right) \:, \label{EoMTransform}
\end{eqnarray}
together with their complex conjugate.

Let $\mathbf{u} = \left(z^{i},\pi^{i}\right)$, then we can write equation (\ref{EoMTransform}) as
\begin{equation}
\frac{d\bar{\mathbf{u}}}{dr} =  \mathcal{J}(\mathbf{u},\bar{\mathbf{u}},r) \:, \label{EoMFinal}
\end{equation}
with
\begin{equation}
\mathcal{J}(u,\bar{u},r) = \left[\begin{array}{c}
    \bar{\pi}^{\bar{i}}\\
    -F(r) \bar{\pi}^{\bar{i}} - \bar{\Gamma}^{\bar{i}}_{\bar{j}\bar{k}}\bar{\pi}^{\bar{j}}\bar{\pi}^{\bar{k}}  + g^{k\bar{i}}\left(G(r)\partial_{k}V_{BH} + H(r)\partial_{k}V\right)
 \end{array}\right] \: . \label{OpJSugra}
\end{equation}
The local existence and uniqueness of the equation
(\ref{EoMFinal}) is established using contraction mapping theorem by
showing the operator $\mathcal{J}$ is a locally Lipshitz.  Let $J
\equiv \left[r_{h},r_{h}+\delta\right] \subset I$ be a real
interval with $\delta$ is a small real constant and $U \subset
T{\mathcal{M}}^{2n_c}$ is an open set. We prove the following
lemma,
\begin{lemma}
Let $\mathcal{J}$ be an operator defined by (\ref{OpJSugra}). If the potentials $V$ and $V_{BH}$ are at least a $C^{2}$ function and satisfying,
\begin{eqnarray}
|\partial_{k}V_{BH}(\tilde{\phi}) - \partial_{k}V_{BH}(\phi)| & \leq & C_{4}|\tilde{\phi}-\phi| \:, \nonumber \\
|\partial_{k}V(\tilde{\phi}) - \partial_{k}V(\phi)| & \leq & C_{5}|\tilde{\phi}-\phi| \:, \label{potentialCondition}
\end{eqnarray}
then the operator $\mathcal{J}$ is a locally Lipshitz with respect to $u$.
\end{lemma}

\begin{proof}

From definition of the operator $\mathcal{J}$ in equation (\ref{OpJSugra}), we have the following estimate,
\begin{eqnarray}
|\mathcal{J}(r,\mathbf{u}(r))|_{U} & \leq & |\pi(r)| + |F(r)||\pi(r)| + \left(|\Gamma(\phi(r),\bar{\phi}(r))|+1\right)\:|\pi(r)|^{2} \nonumber \\
& & \quad + |G(r)|\:|\partial_{k}V_{BH}| + |H(r)|\:|\partial_{k}V|\:. \label{estimateJ1}
\end{eqnarray}
Since $\mathcal{M}^{n_{c}} $ satisfying lemma \ref{LemmaSIM}, then using equation (\ref{KahlerPotential}), we can write (\ref{estimateJ1}) as,
\begin{eqnarray}
|\mathcal{J}(r,\mathbf{u}(r))|_{U} &\leq& \left(|F(r)|+1\right)|\pi(r)| + \left(2\epsilon | \phi(r) |^{3} + C_{1} | \phi(r) | + 1\right)\:|\pi(r)|^{2} \nonumber \\
& & + |G(r)|\:|\partial_{k}V_{BH}| + |H(r)|\:|\partial_{k}V|\:. \label{estimateJ2}
\end{eqnarray}
Since $F(r), G(r), H(r)$ are at least a $C^{2}$ real function, then their value are bounded on any closed interval.
Hence, $|\mathcal{J}(r,\mathbf{u}(r))|_{U}$ is bounded on $J$.

Furthermore, for $\mathbf{u},\tilde{\mathbf{u}}\in U$, we have the following estimate,
\begin{eqnarray}
|\mathcal{J}(r,\tilde{\mathbf{u}}(r))-\mathcal{J}(r,\mathbf{u}(r))|_{U} & \leq & \left(|F(r)|+1\right)|\tilde{\phi}(r)-\phi(r)| + |\pi(r)|^{2}|\tilde{\Gamma}-\Gamma| \nonumber\\
& & + \left(|\Gamma'||\tilde{\pi}(r)+\pi(r)|+1\right)|\tilde{\pi}(r)-\pi(r)|  \nonumber \\
& & + |G(r)||\partial_{k}V_{BH}(\tilde{\phi}(r)) - \partial_{k}V_{BH}(\phi(r))|  \nonumber \\
& & + |H(r)||\partial_{k}V(\tilde{\phi}(r)) - \partial_{k}V(\phi(r))|. \label{JJ}
\end{eqnarray}
Using equation (\ref{KahlerPotential}) and (\ref{potentialCondition}), we can write the equation (\ref{JJ}) as
\begin{eqnarray}
|\mathcal{J}(r,\tilde{\mathbf{u}}(r))-\mathcal{J}(r,\mathbf{u}(r))|_{U} & \leq & \left\{|F(r)| + C_{4}|G(r)|+C_{5}|H(r)| + |\pi(r)|^{2} \right. \nonumber \\
& & \left.\left(2\epsilon |\tilde{\phi}(r)||\tilde{\phi}(r)+\phi(r)| + |\phi(r)|^{2} + C_{1}\right)\right\}|\tilde{\phi}(r)-\phi(r)| \nonumber \\
& & + \left\{\left(2\epsilon |\tilde{\phi}(r)|^{3} + C_{1}|\tilde{\phi}(r)|\right)|\tilde{\pi}(r) + \pi(r)| + 1\right\} \nonumber \\
& &|\tilde{\pi}(r) - \pi(r)|.
\end{eqnarray}
Hence, we have
\begin{equation}
|\mathcal{J}(r,\tilde{\mathbf{u}})-\mathcal{J}(r,\mathbf{u})|_{U} \leq C(|\tilde{\mathbf{u}}|,|\mathbf{u}|)\: |\tilde{\mathbf{u}}-\mathbf{u}|. \label{LipshitzJ}
\end{equation}
Equation (\ref{LipshitzJ}) proves that $\mathcal{J}$ is a locally Lipshitz with respect to $\mathbf{u}$.

\end{proof}

Write equation (\ref{EoMFinal}) in form of integral equation,
\begin{equation}
\bar{\mathbf{u}}(r) = \bar{\mathbf{u}}(r_{h}) + \int_{r_{h}}^{r}\:\mathcal{J}\left(\mathbf{u}(s),\bar{\mathbf{u}}(s),s\right)\:ds \:. \label{IntegralEquation}
\end{equation}
Let
\begin{equation}
X = \{ u \in C(J,T\mathcal{M}^{n_{c}}) : \: \mathbf{u}(r_{h}) = \mathbf{u}_{0}, \: \sup_{r\in J}|\mathbf{u}(r)|\leq M \} \:,
\end{equation}
equipped with the norm,
\begin{equation}
\|\mathbf{u}\|_{X} = \sup_{r\in J}\:|\mathbf{u}(r)|.
\end{equation}
Introducing an operator $\mathcal{K}$ which defined as follow,
\begin{equation}
\mathcal{K}(\mathbf{u}(r)) = \mathbf{u}_{0} + \int_{r_h}^{r}\:ds \mathcal{J}\left(s,\mathbf{u}(s)\right)\:. \label{OpKdefinition}
\end{equation}
Based on equation (\ref{IntegralEquation}), if $\mathbf{u}(r)$ is
a solution of the differential equations (\ref{EoMFinal}), then
$\mathbf{u}(r)$ is a fixed point of operator $\mathcal{K}$. The
existence and the uniqueness of the fixed point of the operator is
guaranteed by the contraction mapping principle.

In the following lemma, we prove that the operator $\mathcal{K}$ is a mapping from $X$ to itself and it is a contraction mapping.
\begin{lemma}
Let $\mathcal{K}$ be an operator defined by equation (\ref{OpKdefinition}). There is a positive constant $\delta$ such that $\mathcal{K}$ is a mapping from $X$ to itself and $\mathcal{K}$ is a contraction mapping on $J = [r_{h},r_{h}+\delta]$.
\end{lemma}

\begin{proof}
From definition of the operator $\mathcal{K}$ in equation (\ref{OpKdefinition}), we have the following estimate,
\begin{eqnarray}
\|\mathcal{K}(\mathbf{u}) \|_{X} & \leq & \|\mathbf{u}_{0}\|_{X} + \sup_{r\in J}  \int_{r_h}^{r}\:ds\:|\mathcal{J}\left(s,\mathbf{u}(s)\right)|\nonumber\\
              & \leq & \|\mathbf{u}_{0}\|_{X} + \sup_{r\in J} \left(|\mathcal{J}(r_h)| + C_{M}|\mathbf{u}|\right) \left(r-r_h\right)\nonumber\\
                     & \leq & \|\mathbf{u}_{0}\|_{X} + \delta\left(\|\mathcal{J}(r_h)\| + M C_{M}\right)\:.
\end{eqnarray}
If we choose,
\begin{equation}
\delta \leq \min\left(\frac{1}{C_{M}},\frac{1}{C_{M}M + \|\mathcal{J}(r_h)\|}\right)\:, \label{TimeSugra}
\end{equation}
then $\mathcal{K}(\mathbf{u})$ is a mapping from $X$ to itself.

Furthermore, for $\mathbf{u},\tilde{\mathbf{u}}\in X$ we have,
\begin{eqnarray}
\|\mathcal{K}(\tilde{\mathbf{u}}) - \mathcal{K}(\mathbf{u})\|_{X} & \leq & \sup_{r\in J}
 \int_{r_h}^{r}\:ds\left| \mathcal{J}\left(s,\tilde{\mathbf{u}}(s)\right) - \mathcal{J}\left(s,\mathbf{u}(s)\right)\right|\nonumber\\
& \leq & \sup_{r\in J} \sup_{0\leq s \leq r} r \left| \mathcal{J}\left(s,\tilde{\mathbf{u}}(s)\right) - \mathcal{J}\left(s,\mathbf{u}(s)\right)\right| \nonumber\\
& \leq & C_{M}\delta\|\tilde{\mathbf{u}} - \mathbf{u}\|_{X}\:.
\end{eqnarray}
Since $\delta$ satisfies the inequality (\ref{TimeSugra}), then the operator $\mathcal{K}$ is a contraction mapping.

\end{proof}

Then, by contraction mapping theorem, the operator $\mathcal{K}$ admits a unique fixed point.
Hence, for each initial value, there exist a unique local solution of the differential equation (\ref{EoMFinal}).

\section{Finite Energy Conditions}
\label{sec:energy}

In order to have global solutions in $J_{\infty}\equiv [r_h,
+\infty)$, we need to consider finite energy conditions which
ensure the existence of such solutions. This aspect is discussed
in this section as follows.\\
\indent Let us first turn to the component of the tensor energy
momentum
\begin{equation}
T_{00} = e^{A-B} g_{i\bar{j}} \phi^i{'}  {\bar{\phi}^{\bar{j}} }{'} + e^{A-2C} V_{\mathrm{BH}} + e^A V \ ,
\end{equation}
whose energy functional is given by
\begin{equation}
E = \int \sqrt{-g^{(3)}}T_{00} ~d^3x = 4\pi \int_{r_h}^{+ \infty}
e^{\frac{1}{2}B + A+C} \left( e^{-B} g_{i\bar{j}} \phi^i{'}
{\bar{\phi}^{\bar{j}} }{'} + e^{-2C} V_{\mathrm{BH}} +  V \right)
dr \ . \label{energysol}
\end{equation}
Defining $J_L \equiv [r_h, L]$ and $J_A \equiv [L, +\infty)$ for
finite and large $L > r_h$. Suppose all functions $A(r)$, $B(r)$,
and $C(r)$ and all scalar fields $\phi(r)$ together with
potentials $(V, V_{\mathrm{BH}})$ are at least $C^2$-function. On
$J_A$ we have $A(r) = - B(r) = {\mathrm{ln}}\Lambda (r)$, and
$C(r) = 2 \ {\mathrm{ln}} r$ where $\Lambda (r)$ is given in
(\ref{Ar}). Moreover, the scalar fields $\phi(r)$ tend to have the
form (\ref{solscalarEOM2}) and the potentials $(V,
V_{\mathrm{BH}})$ become fixed, namely $(V_0, V^0_{\mathrm{BH}})$.
It turns out that in order to have finite energy, on $J_A$ the
scalar potential $V$ should be vanished. So, the energy
(\ref{energysol}) has to be
\begin{equation}
E \le  4\pi \sup_{r\in J_L}|\int_{r_h}^L e^{\frac{1}{2}B + A+C}
\left( e^{-B} g_{i\bar{j}} \phi^i{'} {\bar{\phi}^{\bar{j}} }{'} +
e^{-2C} V_{\mathrm{BH}} +  V \right) dr | + C_0(L, \Sigma,
\phi_0, \bar{\Sigma},  \bar{\phi}_0) \ . \label{energysol1}
\end{equation}
for $C_0 > 0$ and $\Sigma^i$ are the scalar charges. The first
term in the right hand side of (\ref{energysol1}) is also bounded
due to the $C^2$-smoothness of all functions, fields, and the
potentials. This shows that $E$ is
bounded above a positive constant.\\
\indent Therefore, we have proven,
\begin{theorem}
Let ${\mathcal{M}}^{2n_c}$ be a K\"ahler manifold spanned by scalar fields with potential K\"ahler $K$.
 Let $V$ and $V_{BH}$ are scalar and black hole potentials defined by equations (\ref{ScalarPotential}) and (\ref{VBH}), respectively.
 If ${\mathcal{M}}^{2n_c}$ satisfies lemma \ref{LemmaSIM} and $V$ and $V_{BH}$ satisfy condition (\ref{potentialCondition}),
 then for each initial value $\mathbf{u}_{0}$, there is positive constant $\delta$ such that the differential equation (\ref{EoMFinal})
  admits a unique local solution on $[r_{h},r_{h}+\delta]$.
   In particular, this local solution interpolates between two region, namely horizon and asymptotic regions,
    if the scalar potential $V$ vanishes in the asymptotic region.
\end{theorem}

\section*{Acknowledgments}

The work in this paper is supported by Riset KK ITB 2014-2015 and Riset Desentralisasi DIKTI-ITB 2014-2015.


\begin{thebibliography}{99}


\bibitem{Gunara:2013IJMPA}
 B. Gunara, F. Zen, F. Akbar, A. Suroso and Arianto,
  \textit{Some Aspects of Spherical Symmetric Extremal Dyonic Black Holes in $4d$ $N = 1$ Supergravity}, Int. J. Mod. Phys. A {\bf 28}  (2013)   1350084, [arXiv:hep-th/1012.0971].
 
\bibitem{Andrianopoli:2007}
L.~Andrianopoli, R.~D'Auria, S.~Ferrara and M.~Trigiante, 
  \textit{Black-hole attractors in N=1 supergravity}, 
 J. High Energy Phys.  {\bf 0707} (2007) 019 
 [arXiv:hep-th/0703178]. 
 
\bibitem{Ferrara:1995ih} 
  S.~Ferrara, R.~Kallosh and A.~Strominger,
  \textit{N=2 extremal black holes},
  Phys.\ Rev.\ D {\bf 52}, 5412 (1995)
  [hep-th/9508072].
  
\bibitem{Ferrara:1996prd}
S.~Ferrara and R.~Kallosh, 
 \textit{Supersymmetry and attractors}, 
 Phys.\ Rev.\ D {\bf 54} (1996) 1514 
 [hep-th/9602136]. 

 
\bibitem{D'Auria:2001kv}
  R.~D'Auria and S.~Ferrara,
  \textit{On fermion masses, gradient flows and potential in supersymmetric theories},  J. High Energy Phys. {\bf 0105} (2001) 034,
  [arXiv:hep-th/0103153] and the references therein.

\bibitem{Andrianopoli:2001zh}
  L.~Andrianopoli, R.~D'Auria and S.~Ferrara,
  \textit{Supersymmetry reduction of N-extended supergravities in four  dimensions},
  J. High Energy Phys. {\bf 0203} (2002) 025,
  [arXiv:hep-th/0110277] and the references therein.


\bibitem{Belluci:2008prd}
S. Belluci, S. Ferrara, A. Marrani, and A. Yeranyan,
\textit{d = 4 Black Hole Attractors in $N=2$ Supergravity with Fayet-Iliopoulos Terms},
Phys. Rev. D {\bf 77} (2008) 085027, [arXiv:hep-th/0802.0141].


\bibitem{Ferrara:1997npb}
S. Ferrara, G. W. Gibbons, and R. Kallosh,
\textit{Black holes and critical points in moduli space}, Nucl. Phys. B {\bf 500}
 (1997) 75, [arXiv:hep-th/9702103].

\bibitem{Gibbons:1996prl}
G.Gibbons, R. Kallosh, and B. Kol, \textit{Moduli, Scalar Charges,
and the First Law of Black Hole Thermodynamics}, Phys. Rev. Lett.
{\bf 77} (1996) 4992, [arXiv:hep-th/9607108].
\bibitem{KLR}
H. K. Kunduri, J, Lucietti, H. S. Reall, \textit{Near-horizon
symmetries of extremal black holes}, Class. Quant. Grav.  {\bf 24}
(2007) 4169 [arXiv:gr-qc/0705.4214].
\bibitem{AGTZ}
F. T. Akbar, B. E. Gunara, Triyanta, F. P. Zen, \textit{Bosonic Part of 4d N=1 Supersymmetric Gauge Theory with General Couplings: Local Existence}, Adv. Theor. Math. Phys. {\bf 18} (2014)  205, [arXiv:math-ph/1302.4212].

\end{thebibliography}
\end{document}